\documentclass[conference]{IEEEtran}
\IEEEoverridecommandlockouts

\usepackage{cite}
\usepackage{amsmath,amssymb,amsfonts}
\usepackage{bbm}
\usepackage{algorithmic}
\usepackage{graphicx}
\usepackage{textcomp}
\usepackage{enumitem}
\usepackage[normalem]{ulem} 
\usepackage{xcolor}
\newtheorem{theorem}{Theorem}
 
\newtheorem{lemma}[theorem]{Lemma}        
\newtheorem{definition}{Definition}
\newtheorem{claim}{Claim}
\newtheorem{remark}{Remark}
\def\BibTeX{{\rm B\kern-.05em{\sc i\kern-.025em b}\kern-.08em
    T\kern-.1667em\lower.7ex\hbox{E}\kern-.125emX}}

\definecolor{newred}{HTML}{ff382e}
\definecolor{newgreen}{HTML}{549641}
\definecolor{newblue}{HTML}{4c4cfc}
\definecolor{neworange}{HTML}{c98702}

\begin{document}

\title{Error-Correcting Weakly Constrained Codes: \\ Constructions and Achievable Rates}

\author{

\IEEEauthorblockN{Prachi Mishra}
\IEEEauthorblockA{\footnotesize \textit{Dept.\ Electrical Communication Engg.} \\
\textit{Indian Institute of Science}\\
Bengaluru, India\\
prachimishra@iisc.ac.in}
\and

\IEEEauthorblockN{Sidharth Jaggi}
\IEEEauthorblockA{\footnotesize\textit{School of Mathematics} \\
\textit{University of Bristol}\\
Bristol, UK \\
sid.jaggi@bristol.ac.uk}
\and
\IEEEauthorblockN{Navin Kashyap}
\IEEEauthorblockA{\footnotesize\textit{Dept.\ Electrical Communication Engg.} \\
\textit{Indian Institute of Science}\\
Bengaluru, India\\
nkashyap@iisc.ac.in}
\and
\IEEEauthorblockN{Michael Langberg}
\IEEEauthorblockA{\footnotesize\textit{Dept.\ Electrical Engg.} \\
\textit{University at Buffalo}\\
New York, USA \\
mikel@buffalo.edu}

}
\maketitle
\begin{abstract}
We investigate weakly constrained codes, in which specific patterns occur with prescribed frequencies rather than being strictly forbidden as in conventional constrained coding. We propose a capacity-achieving construction of a weakly constrained codebook based on Eulerian cycles. We then obtain, via expurgation, weakly constrained codes with linear minimum distance and positive rate, and analyze the rates achievable. 
Finally, we propose a practical concatenated code construction that supports polynomial-time encoding and decoding. 

\end{abstract}


\section{Introduction}
Constrained coding is fundamental for applications in which specific substrings induce transmission errors. While traditionally integrated into magnetic and optical storage media \cite{magnetic_rec}, constrained codes are also vital for emerging paradigms such as DNA-based data storage, where they mitigate biochemical noise \cite{dna1, dna2}. Furthermore, they are important for modern technologies ranging from energy-harvesting systems \cite{energyharvest} to deep-learning-based detection in resistive random access memory (RRAM) \cite{deeplearning}. Constrained codes are also essential for suppressing patterning effects in inter-symbol interference \cite{isi1, isi2} and optimizing multi-level cell flash memory \cite{flash}.

A constrained code is a set of finite-length sequences in which unwanted patterns are forbidden from appearing as substrings. Such \emph{strong} constraints invariably result in rate loss. A more relaxed notion is that of \emph{weakly constrained codes}, which impose constraints on the frequencies of unwanted patterns without prohibiting them entirely. In \cite{lowerupper1}, Marcus and Roth improved the Gilbert-Varshamov bound on the rate of constrained codes by introducing \emph{soft constraints}, which specify, approximately, the frequencies of occurrence within codewords, of strings from some fixed subset. Further, a construction of weakly constrained codes \cite{weak1} demonstrated significantly improved rates compared to conventional constrained codes.

Weak constraints are particularly relevant in DNA-based data storage. DNA has significant potential as a data storage medium due to its high information density, physical durability, and stability under a wide range of environmental conditions. However, DNA storage imposes several biochemical constraints that must be addressed for reliable data storage. A DNA molecule is modeled as a sequence over the four-letter alphabet $\{A, T, G, C\}$. Certain patterns increase the probability of synthesis and sequencing errors; in particular, long repetitions of the same nucleotide (homopolymer runs) should be limited. Additionally, DNA sequences are typically required to maintain balanced GC-content, i.e., the number of $G$ and $C$ nucleotides is approximately equal to that of $A$ and $T$ nucleotides \cite{dnaerror1, dnaerror2}.

There exists a substantial body of literature on constrained codes \cite{cs1, cs2}; however, the study of weakly constrained codes remains comparatively limited. Different definitions of weak constraints have been proposed in the literature, depending on the types of restrictions imposed on pattern statistics \cite{lower, lowerupper1}. In this work, we adopt the framework proposed in \cite{BS}, which defines weakly constrained systems via a tolerance band of lower and upper bounds on the frequency of all permitted patterns, allowing the tolerance band to narrow in width down to $0$, as the blocklength increases. For both classical (strong) and weak constraints, code sequences can be generated by reading the labels of paths in a labeled directed graph, where the graph structure and labeling function ensure that every path produces a valid constrained sequence. A Markov chain on this graph specifies the associated weak constraints. 

There is an extensive literature on making constrained codes error-resilient by combining them with error-correcting codes; see~\cite[Chapter~9]{cs2} for an overview. In~\cite{BS}, the authors propose a capacity-achieving construction of weakly constrained codes in which information is encoded in constant-weight codewords and subsequently transformed to obtain weakly constrained sequences. In~\cite{our}, a modification of this scheme is proposed to provide error resilience for weak constraints defined by Markov chains on first-order deBruijn graphs (i.e., constraints on length-two patterns). 
Lower bounds were established in \cite{lowerupper1} on the attainable asymptotic rate of weakly constrained codes for a given relative minimum distance, but these results are non-constructive. They prove the existence of codes where codewords possess specific empirical statistics, but they do not provide an algorithmic mechanism for codeword generation. In this work, we bridge this gap by introducing a structural framework that models codewords as Eulerian cycles on a multigraph.

We first propose an algorithm for constructing weakly constrained codebook via Eulerian cycles, which asymptotically achieve the capacity of the constrained system. To incorporate error-correction capability, we establish, via an \emph{expurgation} argument, the existence of subcodes that attain both linear minimum distance and a positive rate for finite blocklengths. We also analyze the rates achievable through this construction, both at finite blocklengths and asymptotically in blocklength.

However, the expurgation-based approach is inherently non-constructive and does not yield efficient encoding or decoding procedures. To address this limitation, we introduce a \emph{concatenated construction} that employs the expurgated code as an inner code, combined with a Reed-Solomon outer code. By adopting a specific scaling rule for the inner blocklength, the resulting scheme achieves polynomial-time encoding and decoding while preserving the weak constraints and further amplifying the distance guarantees of the expurgated code.


The remainder of this manuscript is structured as follows: Section \ref{sec:preliminaries} establishes the preliminaries required to explain our construction. Section \ref{sec:weak_codes_randomwalks} provides a capacity-achieving construction of weakly constrained codes via Eulerian cycles and analyzes the resulting codebook size. Section \ref{sec:error_correcting_wcc} introduces the expurgation-based approach for error-correcting codes and derives the corresponding asymptotic and finite length rate bounds. Section   \ref{sec:concatenation} presents the efficient concatenated construction that enables polynomial-time encoding and decoding. Finally, Section \ref{sec:conclusion} provides concluding remarks.

\section{Preliminaries}
\label{sec:preliminaries}
\subsection{Labeled Graphs}


Let $\Sigma$ be a finite alphabet and let $\Sigma^*$ denote the set of all finite-length sequences over $\Sigma$. A labeled directed graph over $\Sigma$ is defined as a tuple $G = (V, E, L)$, where $V$ and $E$ are finite sets of vertices and edges, respectively, and $L: E \to \Sigma$ is a labeling function that assigns a symbol to each edge. $G$ is said to be \emph{deterministic} if for any vertex $v\in V$ and any symbol $a\in \Sigma$, there is at most one edge originating from $v$ such that $L(e)=a$.

For an edge $e \in E$, let $\sigma(e)$ and $\tau(e)$ denote its initial and terminal vertices, respectively. For any vertex $v \in V$, the sets of outgoing and incoming edges are defined as
\[
E_{\text{out}}(v) = \{e \in E : \sigma(e) = v\}, \quad
E_{\text{in}}(v) = \{e \in E : \tau(e) = v\}.
\]
The out-degree and in-degree of $v$ are given by $d_{\text{out}}(v) = |E_{\text{out}}(v)|$ and $d_{\text{in}}(v) = |E_{\text{in}}(v)|$, respectively.

Let $\Gamma$ denote the set of all finite-length paths in $G$. A path $\gamma \in \Gamma$ of length-$\ell$ in $G$ is a sequence of edges $(e_1, e_2, \ldots, e_\ell)$ such that $\tau(e_i) = \sigma(e_{i+1})$ for all $1 \le i < \ell$. The labeling function naturally extends to paths by concatenation of edge labels; that is,
\[
L(\gamma) = L(e_1)L(e_2)\cdots L(e_\ell) \in \Sigma^*.
\]
A cycle is a path whose initial and terminal vertices coincide.

A directed graph $G$ is said to be irreducible if for every pair of vertices $u, v \in V$, there exists a path from $u$ to $v$. The graph is aperiodic if the greatest common divisor of the lengths of all cycles in $G$ is one. A graph is primitive if there exists an integer $N_G > 0$ such that for every pair of vertices $u, v \in V$, there exists a path of length $N_G$ from $u$ to $v$. A labeled graph is called lossless if any two distinct paths with the same initial and terminal vertices produce distinct label sequences.

An Eulerian cycle (resp.\ path) in $G$ is a cycle (resp.\ path) that traverses each edge exactly once. A directed graph contains an Eulerian cycle if and only if it is irreducible and balanced, i.e., $d_{\text{in}}(v) = d_{\text{out}}(v)$ for all $v \in V$. Similarly, a directed graph contains an Eulerian path from $u$ to $v$ if and only if it is irreducible and semi-balanced, i.e.,
\[
d_{\text{out}}(u) - d_{\text{in}}(u) = 1, \quad
d_{\text{in}}(v) - d_{\text{out}}(v) = 1,
\]
and $d_{\text{in}}(w) = d_{\text{out}}(w)$ for all $w \notin \{u, v\}$.

\subsection{Markov Chains}

To describe the statistical behavior of paths in the graph $G$, we associate a Markov chain with it. A Markov chain on $G$ is specified by a probability mass function $P: E \to [0,1]$ on the edge set, satisfying
\[
\sum_{e \in E} P(e) = 1.
\]
This induces a distribution $\pi$ on the vertex set $V$, defined by
\[
\pi(u) = \sum_{e:\sigma(e)=u} P(e), \quad u \in V.
\]
The Markov chain is said to be \emph{stationary} if, for all $u \in V$,
\[
\pi(u) = \sum_{e:\tau(e)=u} P(e).
\]
The set of all stationary Markov chains on $G$ is denoted by $\mathcal{M}_s(G)$.

It is often convenient to describe the chain in terms of transition probability matrix $P_{\mathrm{trans}}={p(\cdot|\cdot)}$. For vertices $u, v \in V$, 
\[
p(v|u) = \frac{P(e)}{\pi(u)}, \quad \text{if } \pi(u) > 0,
\]
where $e = (u,v)$. A Markov chain is said to be \emph{reversible} if 
\[
\pi(u)p(v|u) = \pi(v)p(u|v), \quad \forall u, v \in V.
\]

The entropy rate of a stationary Markov chain $P$ is given by
\[
H(P) = - \sum_{e \in E} P(e) \log_2 \left( \frac{P(e)}{\pi(\sigma(e))} \right).
\]

The empirical Markov chain of a path $\gamma =(e_1, e_2, \ldots, e_\ell)$ is defined as
\[
P_\gamma(e) \triangleq \frac{1}{\ell} \left| \{ i \in \{1, \ldots, \ell\} : e_i = e \} \right|, \quad e \in E.
\]

Let $\lambda_1, \lambda_2, \dots, \lambda_{|V|}$ be the eigenvalues of $P_{\mathrm{trans}}$. For both reversible and non-reversible Markov chains, the absolute spectral gap is $\gamma^* := 1 - \max \{| \lambda_i | : \lambda_i \neq 1 \}$, while for reversible chains, the spectral gap is defined as 
\begin{equation}
\label{eq:spectralgap}
    \gamma := 1 - \max \{ \lambda_i : \lambda_i \neq 1 \}\text{.}
\end{equation}
In the case of non-reversible chains, the pseudo spectral gap is defined as
\begin{equation}
\label{eq: pseudogap}
    \gamma_{ps} := \max_{k \ge 1} \left\{ \frac{\gamma((P_{\mathrm{trans}}^*)^k P_{\mathrm{trans}}^k)}{k} \right\},
\end{equation}
where $P_{\mathrm{trans}}^*$ is the adjoint operator defined by the entries $p^*(u|v) = \frac{\pi(u)p(v|u)}{\pi(v)}$; see \cite{bern} for details.

For a positive integer $n$, the chain $P$ is said to be \emph{$n$-integral} if $nP(e)$ is an integer for all $e \in E$. 

\subsection{Random Walks and Hamming Distance}
A random walk on the graph $G$ according to a Markov chain $P$, is a stochastic process $\{V_t\}_{t \geq 1}$ on the vertex set $V$, where the initial state $V_1 = v_1$, and for all $t \geq 1$,
\[
P(V_{t+1} = v \mid V_t = u) = p(v|u), \quad u, v \in V.
\]
This process induces a sequence of edges $(e_1, e_2, \ldots, e_n)$, where each $e_t = (V_t, V_{t+1})$, forming a path in $G$. The corresponding label sequence
$x = (x_1, x_2, \ldots, x_n)$, where $ x_t = L(e_t),$ represents a valid sequence generated by the graph.

The Hamming distance between two such sequences $\mathbf{x} = (x_1, \ldots, x_n)$ and $\mathbf{y} = (y_1, \ldots, y_n)$ is defined as
\[
d_H(\mathbf{x}, \mathbf{y}) = \sum_{k=1}^n \mathbb{I}(x_k \neq y_k),
\]
where $\mathbb{I}(\cdot)$ denotes the indicator function. Its normalized version, $\delta_H(\mathbf{x}, \mathbf{y}) = \frac{1}{n}d_H(\mathbf{x}, \mathbf{y})$, denotes the \emph{relative Hamming distance}.

\subsection{Bernstein's Inequality for Markov Chains}
\label{eq:bernstein}
To analyze deviations of empirical quantities, we use concentration inequalities for Markov chains from \cite{bern}. Let $\{X_i\}_{i=1}^n$ be a Markov chain on the edge set $E$. Assume first that the chain is initialized by the stationary distribution $\pi$, i.e., $X_1 \sim \pi $. Let $L^2(\pi)$ be the Hilbert space of complex-valued measurable functions on $E$ that are square-integrable with respect to $\pi$.

For a function $f \in L^2(\pi)$ satisfying
\[
|f(e) - \mathbb{E}_\pi[f]| \le C \quad \forall e \in E,
\]
define
\[
D = \sum_{i=1}^n f(X_i), \quad V_f = \mathrm{Var}_\pi(f).
\]
For a reversible Markov chain with spectral gap $\gamma$
\begin{equation}
\label{eq:ber_rev}
P_\pi\big(|D - \mathbb{E}_\pi(D)| \ge t\big)
\le 2 \exp\left(
-\frac{t^2 \gamma}{4n V_f + 10 t C}
\right)
\end{equation}
and for a non-reversible Markov chain with pseudo-spectral gap $\gamma_{\mathrm{ps}}$
\begin{equation}
\label{eq:ber_non-rev}
    P_\pi\big(|D - \mathbb{E}_\pi[D]| \ge t\big)
\le 2 \exp\left(
-\frac{t^2 \gamma_{\mathrm{ps}}}{8(n + 1/\gamma_{\mathrm{ps}})V_f + 20tC}
\right).
\end{equation}

If the chain is initialized according to an arbitrary distribution $q$ on $E$, then
\begin{equation}
\label{eq:ber_nonsta}
  P_q\big(|D- \mathbb{E}_\pi[D]| \ge t\big)
\le \sqrt{N_q} \, \big[P_\pi(|D - \mathbb{E}_\pi[D]| \ge t)\big]^{1/2},  
\end{equation}
where
\[
N_q = \sum_{e \in E} \frac{q(e)^2}{\pi(e)}.
\]

\subsection{Large Deviation Principle for Markov Chains}
\label{subsec:ldp_prelim}

To characterize the asymptotic deviations of empirical quantities, we utilize the Large Deviation Principle (LDP) for Markov chains \cite{ldp}. Let $\mathcal{M}_1(V)$ denote the space of all probability measures on the alphabet $V$. For any measure $Q \in \mathcal{M}_1(V \times V)$, its marginals $Q_1, Q_2 \in \mathcal{M}_1(V)$ are defined as
\begin{equation*}
    Q_1(i) = \sum_{j \in V} Q(i, j), \quad Q_2(i) = \sum_{j \in V} Q(j, i),
\end{equation*}
for all $i \in V$. A measure $Q$ is said to be \emph{shift-invariant} (or stationary) if its marginals coincide, i.e., $Q_1(i) = Q_2(i)$ for all $i \in V$. The set of all such stationary Markov measures are denoted by $\mathcal{M}_s(V \times V) \subset \mathcal{M}_1(V \times V)$.

Let $\{X_k\}$ be a Markov chain on $G=(V,E)$ with transition matrix $P_{\mathrm{trans}} = (p(v|u))_{u,v \in V}$ having strictly positive entries $p(v|u) > 0$, and a unique stationary distribution $\pi$. We define the \emph{pair empirical measure} $\nu_{n} \in \mathcal{M}_1(V \times V)$ as
\begin{equation*}
    \nu_{n}(u,v) \triangleq \frac{1}{n} \sum_{k=1}^{n} \mathbb{I}(X_k = u, X_{k+1} = v), \quad u,v \in V.
\end{equation*}
The sequence of laws $\mu_{n}(\cdot) \triangleq P_{\pi}(\nu_{n} \in \cdot)$ satisfy the LDP on $\mathcal{M}_s(V \times V)$ with rate $n$ and rate function $I(Q)$, that is, 
for any Borel set $A \subseteq \mathcal{M}_s(V \times V )$, 
\begin{multline}
    -\inf_{Q \in A^\circ} I(Q) \le \liminf_{n \to \infty} \frac{1}{n} \log \mu_n(A) \\
    \le \limsup_{n \to \infty} \frac{1}{n} \log \mu_n(A) \le -\inf_{Q \in \overline{A}} I(Q).
\end{multline}
The rate function $I(Q)$ is defined as 
\begin{equation}
    I(Q) = \sum_{(u,v) \in E} Q(u,v) \log \frac{Q(v|u)}{p(v|u)},
\end{equation}
where $Q(v|u) = Q(u,v)/Q_1(u)$.

\subsection{Whittle's Formula}
\label{whittl}
Let $(x_1, x_2, \dots, x_{n+1})$ be a sequence of states and let $F = [f_{ij}]$ be an $s \times s$ matrix of non-negative integers representing the transition counts, where $s$ denotes the number of states and $f_{ij} = |\{m \in \{1, \dots, n\} : x_m = i, x_{m+1} = j\}|$. Let $f_{i\cdot} = \sum_{j} f_{ij}$ and $f_{\cdot i} = \sum_{j} f_{ji}$ denote the total number of exits from and entries into state $i$, respectively.

If the sequence satisfies the condition that, for a specific pair of states $u, v$, $f_{i\cdot} - f_{\cdot i} = \delta_{iu} - \delta_{iv}$ for all $i = 1, \dots, s$, then the number of such sequences starting at $u$ and ending at $v$, denoted by $N_{uv}^{(n)}(F)$, is given by
\begin{equation}
N_{uv}^{(n)}(F) = \frac{\prod_{i} f_{i\cdot}!}{\prod_{i,j} f_{ij}!} F_{vu}^*
\end{equation}
where $F_{vu}^*$ is the $(v, u)$-th cofactor of the matrix $F^* = \{f_{ij}^*\}$ with components defined as
\begin{equation}
\label{eq:cofactormatrix}
f_{ij}^* = \begin{cases} 
\delta_{ij} - f_{ij}/f_{i\cdot} & \text{if } f_{i\cdot} > 0 \\ 
\delta_{ij} & \text{if } f_{i\cdot} = 0 
\end{cases}
\end{equation}
Whittle's formula thus enables the exact combinatorial counting of sequences with a prescribed transition profile specified by $F$; see \cite{whittle1,whittle2} for details.

\subsection{Constrained Systems}
\label{condtrained_sys}
A \emph{constrained system} $S(G)$ is the set of all finite-length words that can be generated by reading the labels of valid paths in $G$, i.e.,
\[
S(G)=\{w\in\Sigma^*:\exists\,\gamma \text{ in } G \text{ such that } w=L(\gamma)\}.
\]
A \emph{constrained code} is any subset of $S(G)$. The \emph{capacity} of the constrained system $S(G)$ is defined as
\begin{equation} \label{def:capS}
\mathrm{cap}(S(G))
=\lim_{\ell\to\infty}\frac{1}{\ell}\log_2|S(G)\cap\Sigma^\ell|.
\end{equation}
Capacity is also formulated as 
\[
\mathrm{cap}(S(G))=\sup_{P} H(P),
\]
where the supremum is taken over all stationary Markov chains $P$ on $G$; see \cite{BS}.

For every constrained system, there exists a lossless, primitive labeled graph $\widehat{G}$ that preserves the system capacity \cite{BS}. Specifically, $S(\widehat{G}) \subseteq S(G)$ and 
\begin{equation*}
    \mathrm{cap}(S(\widehat{G})) = \mathrm{cap}(S(G)).
\end{equation*}
Thus, we can assume, as we henceforth do, that the underlying graph $G$ is lossless and primitive.

\subsection{Weakly Constrained Systems}
\label{sec:weakly_constrained_def}
Let $F \subseteq E$ be a subset of the edges of $G$, and let 
$\hat{\mathbf{r}} = {(\hat{r}_f)}_{f \in F} \in {[0,1]}^{F}$ be a vector of target frequencies assigned to the edges in $F$. Also, let $\epsilon: \mathbb{N} \to \mathbb{R}_{\ge 0}$ be a tolerance function. The $(\hat{\mathbf{r}}, \epsilon)$-weakly constrained system is defined as the set of all words in $S(G)$ generated by paths $\gamma$ whose empirical statistics $P_{\gamma}$ are close to the target $\hat{\mathbf{r}}$
\begin{equation*}
\scalebox{0.95}{%
    $S_{\hat{\mathbf{r}}, \epsilon}(G) \triangleq
    \left\{
    \mathbf{c} \in S(G) :
    \begin{array}{l}
    \exists \gamma \in \Gamma \text{ s.t. } L(\gamma)=\mathbf{c},\\ \forall f \in F, 
    | P_{\gamma}(f) - \hat{r}_{f}| \le \epsilon(|\gamma|)
    \end{array}
    \right\}$%
} \text{.}
\end{equation*}
The capacity of $S_{\hat{\mathbf{r}}, \epsilon}(G)$ is defined analogously to \eqref{def:capS}. When $\epsilon(n)=\Theta(1/n)$, the capacity is equal to $H(\hat{P})$, where $\hat{P}$ is the maxentropic Markov chain on $G$ subject to the condition $\hat{P}(f) = \hat{r}_f$ for all $f \in F$; see \cite[Remark~1]{BS}. In this case, we define $\mathrm{cap}(S_{\hat{\mathbf{r}},\epsilon}(G)) := H(\hat{P})$.




\section{Weakly Constrained Codes via Eulerian Cycles}
\label{sec:weak_codes_randomwalks}
A weakly constrained code, for us, is any subset of $S_{\hat{\mathbf{r}},\epsilon}(G)$ as defined in Section~\ref{sec:weakly_constrained_def}. In order to construct weakly constrained codes, we consider $n$-integral Markov chains $P$ that are arbitrarily close to the maxentropic chain $\hat{P}$. To be precise, we take a sequence of $n$-integral stationary Markov chains $P^{(n)}$, $n \in \mathbb{N}$, such that ${\|P^{(n)} - \hat{P}\|}_{\infty} = \max_{e \in E} |P^{(n)}(e) - \hat{P}(e)| = \mathcal{O}(1/n)$; such a sequence always exists \cite{BS},\cite{TER09}. Moreover, we can take $P^{(n)}(e) > 0$ for all $e \in E$. Set $\mathbf{r}^{(n)} := P^{(n)}$, i.e., $r^{(n)}_e = P^{(n)}(e)$ for all $e \in E$. Then, ${\|P^{(n)} - \hat{P}\|}_{\infty} = \mathcal{O}(1/n)$ implies that $|r^{(n)}_f - \hat{r}_f| = \mathcal{O}(1/n)$ for all $f \in F$. Consequently, for $\epsilon(n) = \Theta(1/n)$, we have $S_{\mathbf{r}^{(n)},0}(G) \subseteq S_{\hat{\mathbf{r}},\epsilon}(G)$, where $S_{\mathbf{r}^{(n)},0}(G)$ consists of all words in $S(G)$ generated by paths $\gamma$ in $G$ in which each edge $e \in E$ (including those not in $F$) appears with frequency exactly $P^{(n)}(e)$. From this, we conclude that weakly constrained codes for $S_{\hat{\mathbf{r}},\epsilon}(G)$ can be obtained from code constructions for the system $S_{\mathbf{r}^{(n)},0}(G)$. 

Fix a $\xi > 0$, which we can take to be arbitrarily small, and let $n_0$ be such that ${\|P^{(n_0)} - \hat{P}\|}_{\infty} < \xi$ and $P^{(n_0)}(e) > 0$ for all $e \in E$. Given a blocklength $n$ that is an integer multiple of $n_0$ such that $n>n_0$, we construct, from the \emph{primal graph} $G$ that defines $S(G)$, a multigraph $G_n = (V, E_n)$ by replicating each edge $e \in E$ exactly $nP^{(n_0)}(e)$ times. The stationarity of $P^{(n_0)}$ ensures that $G_n$ is balanced and therefore, by definition, is an Eulerian graph. Any Eulerian cycle in $G_n$ corresponds to a length-$n$ path in which each edge $e$ appears with frequency exactly $P^{(n_0)}(e)$. The resulting label sequences form a valid weakly constrained codebook, denoted by $\mathcal{C}_{\mathrm{wcc}}$. While every codeword in $\mathcal{C}_{\mathrm{wcc}}$ satisfies the weak constraints, we restrict our attention to a structured subset $\mathcal{C}_{\mathrm{pool}} \subset \mathcal{C}_{\mathrm{wcc}}$ that enables efficient analysis of the system properties. In the sequel, to simplify notation, we drop the superscript from $P^{(n_0)}$.

\subsection{Codebook Construction}
Fix a parameter $\alpha \in (0,1)$ such that $\frac{1}{1-\alpha}$ is an integer, and let $n = z \cdot \frac{n_0}{1-\alpha}$, where $z$ is some positive integer. Thus, $n$ is an integer multiple of $n_0$, so that $P$ is also $n$-integral. Note that by letting $z \to \infty$, we can have $n \to \infty$.
The codebook $\mathcal{C}_{\mathrm{pool}}\subset \Sigma^n$ comprises label sequences of length-$n$ Eulerian cycles formed by a unique completion of a prefix of length $n' = \alpha n$ that lies in a set of admissible prefixes. Note that $n'P(e)$ is an integer for all $e \in E$.  We define the set of \emph{admissible prefixes}, denoted by $\mathcal{W}$, as the set of all paths $\mathbf{w}$ of length $n'$ in $G$ starting at a fixed root-vertex $v_{\mathrm{root}}$, that are $\zeta$-typical with respect to $P$, i.e.,
\begin{equation*}
    \mathcal{W} \triangleq \left\{ \mathbf{w} \in E^{n'} : \left| \frac{S(e)}{n'} - P(e) \right| < \zeta, \forall e \in E \right\} \text{,}
\end{equation*}
where $S(e)$ denotes the number of times edge $e$ is traversed in $\mathbf{w}$ and $0 < \zeta < (\frac{1-\alpha}{\alpha})P_{\min}$, with $P_{\min} = \min_{e \in E} P(e)$.

\begin{lemma}
\label{lem:eulerian_completion}
For any admissible prefix $\mathbf{w} \in \mathcal{W}$, there exists a suffix $\tilde{\mathbf{w}}$ such that the combined path $\mathbf{w} || \tilde{\mathbf{w}}$ forms an Eulerian cycle in the multigraph $G_n$.
\end{lemma}
\begin{IEEEproof}
For any admissible prefix $\mathbf{w} \in \mathcal{W}$, with $\mathbf{w} = (e_{1}, \dots, e_{n'})$,  terminating at $v_{\mathrm{end}}$, we define a residual graph $G_{\mathrm{res}} = (V, E_{\mathrm{res}})$ as the subgraph obtained by removing the edges traversed in $\mathbf{w}$ from $G_n$. The multiplicity of each edge $e$ in $G_{\mathrm{res}}$ is
\[
R(e) = nP(e) - S(e).
\] 
$G_{\mathrm{res}}$ is semi-balanced as $G_n$ is balanced and $\mathbf{w}$ satisfies $d_{\mathrm{in}}(v) = d_{\mathrm{out}}(v)$ for all $v \notin \{v_{\mathrm{root}}, v_{\mathrm{end}}\}$. Moreover, the definition of $\mathcal{W}$ ensures that
\begin{equation}
\label{eq:prefixlength}
    S(e) < n'(P(e) + \zeta) < \alpha n \left( P(e) + \frac{1-\alpha}{\alpha} P_{\min} \right) \le nP(e) \text{,}
\end{equation}
 implying the \emph{feasibility condition} i.e., $S(e) < nP(e)$, i.e., $R(e) > 0$ for all $e$ with $P(e) > 0$. Therefore, the structural constraints imposed on $\mathbf{w}$ ensure that $G_{\mathrm{res}}$ is irreducible, thereby ensuring the existence of an Eulerian path from $v_{\mathrm{end}}$ back to $v_{\mathrm{root}}$. This path corresponds to the required suffix $\tilde{\mathbf{w}}$.
\end{IEEEproof}
We further define a deterministic map, $\Phi : \mathcal{W} \to E^{n-n'}$, which maps each admissible prefix $\mathbf{w} \in \mathcal{W}$ to an Eulerian path $\mathbf{w}_c = \Phi(\mathbf{w})$. Here, $\mathbf{w}_c = (e_{n'+1}, \dots, e_n)$ is the lexicographically first Eulerian path in $G_{\mathrm{res}}$ from $v_{\mathrm{end}}$ back to $v_{\mathrm{root}}$.

The codebook $\mathcal{C}_{\mathrm{pool}}$ is then defined as 
\begin{equation}
    \mathcal{C}_{\mathrm{pool}} = \{ L(\mathbf{w} \parallel \Phi(\mathbf{w})) : \mathbf{w} \in \mathcal{W} \} \text{,}
\end{equation}
where $\mathbf{w}$ corresponds to the admissible prefix of length $n'$ and $\Phi(\mathbf{w})$ provides the structural completion of the Eulerian cycle for maintaining the weak constraints.

While $\mathcal{C}_{\text{pool}}$ is a fixed set, we associate with each codeword $\mathbf{c} \in \mathcal{C}_{\text{pool}}$ a probability $P(\mathbf{c})$, defined as the probability of its prefix $\mathbf{w}$ being generated by a random walk of length-$n'$ on $G$ governed by the Markov chain $P$. Here, random walks are used strictly for analysis and not for the codebook construction. For this construction, we assume $G$ is deterministic; therefore, there is a one-to-one correspondence between the prefix $\mathbf{w}$ and its label sequence $L(\mathbf{w})$. Furthermore, since the structural completion $\Phi(\mathbf{w})$ is deterministic, the prefix $\mathbf{w}$ uniquely determines the codeword $\mathbf{c}$. This allows us to use $P(\mathbf{c})$ and $P(\mathbf{w})$ interchangeably; therefore, with a slight abuse of notation, we sometimes write $P(\mathbf{w} \in \mathcal{C}_{\text{pool}})$ to denote the total probability mass of the pool. This probabilistic perspective serves as a tool to characterize the total probability mass of the pool and to analyze the distance properties of the resulting sequences.

\begin{theorem}[Reliability of Construction]
\label{thm:reliability}
For the design parameter $\alpha$ and typicality parameter $\zeta$ satisfying $0 < \zeta < (\frac{1-\alpha}{\alpha})P_{\min}$, the total probability mass of $\mathcal{C}_{\mathrm{pool}}$
\begin{equation*}
    P(\mathcal{C}_{\mathrm{pool}}) = \sum_{\mathbf{c} \in \mathcal{C}_{\mathrm{pool}}} P(\mathbf{c}) \ge 1 - |E|\sqrt{\frac{2}{\pi_{\max}}} \exp(-n \eta_{\mathrm{typ}}) \text{,}
\end{equation*}
where $\pi_{\max} = \max_{v \in V} \pi(v)$ and the reliability exponent $\eta_{\mathrm{typ}}$ is defined as
\begin{equation*}
\eta_{\mathrm{typ}} = \begin{cases} 
\frac{\alpha\zeta^2\gamma}{8V_f + 20\zeta} & \text{if } P \text{ is reversible} \\
\frac{\alpha\zeta^2\gamma_{ps}}{16V_f(1+\gamma_{ps}^{-1}) + 40\zeta} & \text{if } P \text{ is non-reversible}
\end{cases}
\end{equation*}
Here, $V_f = P_{\min}(1-P_{\min})$ with $P_{\min} = \min_{e \in E} P(e)$, and $\gamma$ (resp. $\gamma_{\mathrm{ps}}$) is the spectral gap (resp. pseudo-spectral gap) of the Markov chain $P$, as defined in \eqref{eq:spectralgap} (resp. \eqref{eq: pseudogap}). In particular, $\lim_{n \to \infty} P(\mathcal{C}_{\text{pool}}) = 1$.
\end{theorem}

\begin{IEEEproof}
   The total probability mass $P(\mathcal{C}_{\mathrm{pool}}) = \sum_{\mathbf{c} \in \mathcal{C}_{\mathrm{pool}}} P(\mathbf{c})$ is equivalent to the probability that a length-$n'$ random walk on $G$ generates an admissible prefix, $P(\mathbf{w} \in \mathcal{W})$. Therefore, $P(\mathcal{C}_{\mathrm{pool}}) = P(\mathbf{w} \in \mathcal{T}_{n'}(\zeta))$, for $0 < \zeta < (\frac{1-\alpha}{\alpha})P_{\min}$.

To bound $P(\mathbf{w}\in \mathcal{T}_{n'}(\zeta))$, we first consider a random walk $\mathbf{w} = (e_1, \dots, e_{n'})$ initialized according to its stationary distribution $\pi$. For each step $k \in [n']$, the expectation of the edge indicator function $f_e(e_k) = \mathbb{I}(\{e_k = e\})$ is
\begin{equation*}
    \mathbb{E}_{\pi}[f_e(e_k)] = P_{\pi}(e_k = e) = \pi(\sigma(e))\,p(\tau(e)\mid\sigma(e)) = P(e),
\end{equation*}
and variance $V_f := \mathrm{Var}_{\pi}(f_e) = P(e)(1 - P(e))$. 
By linearity of expectation, $\mathbb{E}_{\pi}[S(e)] = n'P(e)$. 

Applying the union bound over the edge set $E$, we have
\begin{equation*}
    P(\mathbf{w} \notin \mathcal{T}_{n'}(\zeta)) \le \sum_{e \in E} P\left( \left| \frac{S(e)}{n'} - P(e) \right| \ge \zeta \right).
\end{equation*}

When $P$ is a non-reversible chain, we apply the Bernstein inequality for non-reversible Markov chains with stationary initialization given in \eqref{eq:ber_non-rev}, with deviation $t = n'\zeta$ and $C=1$, we get 
\begin{equation*}
    P_{\pi}\left( \left| \frac{S(e)}{n'} - P(e) \right| \ge \zeta \right) \le 2 \exp \left( - \frac{n' \zeta^2 \gamma_{\mathrm{ps}}}{8V_f(1+\frac{1}{n' \gamma_{ps}}) + 20\zeta} \right).
\end{equation*}
Since $\gamma_{\mathrm{ps}}$ is a fixed property of the Markov chain, we have $(1+\frac{1}{n' \gamma_{ps}}) < (1+\frac{1}{\gamma_{ps}})$. Hence, we obtain the bound
\begin{equation*}
    P_{\pi}\left( \left| \frac{S(e)}{n'} - P(e) \right| \ge \zeta \right) \le 2 \exp \left( - \frac{n' \zeta^2 \gamma_{\mathrm{ps}}}{8V_f(1+\frac{1}{\gamma_{ps}}) + 20\zeta} \right).
\end{equation*}

To account for the actual construction initialized at $v_{\mathrm{root}}$, we apply the non-stationary version of Bernstein's inequality \eqref{eq:ber_nonsta}, which yields
\begin{equation*}
\scalebox{0.99}{$
P\!\left( \left| \frac{S(e)}{n'} - P(e) \right| \ge \zeta \right)
\le
\sqrt{\frac{2}{\pi(v_{\mathrm{root}})}}
\exp\!\left(
- \frac{n' \zeta^2 \gamma_{\mathrm{ps}}}
{16V_f \left(1+\frac{1}{\gamma_{\mathrm{ps}}}\right) + 40\zeta}
\right)
$}
\end{equation*}
Similar analysis for reversible Markov chains with $t = n'\zeta$ and $C=1$, gives
\begin{equation*}
P\!\left( \left| \frac{S(e)}{n'} - P(e) \right| \ge \zeta \right)
\le
\sqrt{\frac{2}{\pi(v_{\mathrm{root}})}}
\exp\!\left(
- \frac{n' \zeta^2 \gamma}
{8V_f + 20\zeta}
\right)
\end{equation*}
By selecting $v_{\mathrm{root}}$ such that $\pi(v_{\mathrm{root}}) = \max_{v \in V} \pi(v) = \pi_{\max}$ to minimize the prefactor, we obtain the final reliability bound
\begin{equation*}
    P(\mathcal{C}_{\mathrm{pool}}) \ge 1 - |E|\sqrt{\frac{2}{\pi_{\max}}} \exp(-n \eta_{\mathrm{typ}}).
\end{equation*}
\end{IEEEproof}

\subsection{Codebook Size Analysis}
\label{subsec:combinatorial}

To characterize the information-carrying capacity of $ \mathcal{C}_{\mathrm{pool}}$ at finite blocklengths, we enumerate the valid Eulerian cycles in the pool using Whittle's formula (see Section~\ref{whittl}).  Since the codebook $\mathcal{C}_{\mathrm{pool}}$ is in one-to-one correspondence with the set of admissible prefixes $\mathcal{W}$, i.e. all walks $\mathbf{w}$ with empirical edge frequencies close to $P$, we lower bound its size by enumerating paths of length $n' = \alpha n $ whose transition counts are exactly $n'P(e)$, for all $e \in E$. Substituting these parameters, the number of such sequences starting at $v_{\mathrm{root}}$ and ending at $v \in V$ is
\begin{equation*}
    N_{v_{\mathrm{root}}, v}^{(n')}(P) = \frac{\prod_{i \in V} (n' \pi(i))!}{\prod_{e \in E} (n' P(e))!} F_{v, v_{\mathrm{root}}}^*(P) ,
\end{equation*}
where $n'\pi(i)$ is the total number of times the walk exits a vertex $i\in V$. The term $F_{v, v_{\mathrm{root}}}^*(P)$ is the $(v, v_{\mathrm{root}})$-th cofactor of the matrix $F^* = \{f_{ij}^*\}$ with components defined as
\begin{equation*}
f_{ij}^* =  \delta_{ij} - P(i,j)/\pi(i) 
\end{equation*}

We apply the upper and lower bounds for Stirling's approximation, valid for all $n'$,
\[
\sqrt{2\pi n} (n/e)^n e^{1/(12n+1)} < n! < \sqrt{2\pi n} (n/e)^n e^{1/(12n)}
\]
By taking the lower bound of the numerator and the upper bound of the denominator, the log-count for paths ending at vertex $v$ can be expressed as
\begin{equation*}
    \log_{2} N_{v_{\mathrm{root}}, v}^{(n')} \ge n' H(P) + \frac{|V| - |E|}{2} \log_{2}(2\pi n')+ \Delta(n')+ C_v ,
\end{equation*}
where $C_v$ is a vertex-dependent constant defined as
\begin{equation*}
    C_v = \frac{1}{2} \left( \sum_{u \in V} \log_{2} \pi(u) - \sum_{e \in E} \log_{2} P(e) \right) + \log_{2} F_{v, v_{\mathrm{root}}}^* .
\end{equation*}
and 
\[
\Delta(n')= \log_2(e) \left( \sum_{i \in V} \frac{1}{12n'\pi(i) + 1} - \sum_{e \in E} \frac{1}{12n'P(e)} \right).
\]
The total pool size $|\mathcal{C}_{\mathrm{pool}}|$ is bounded by the summing over all terminal states $v \in V$,
\begin{equation}
\label{eq:poolsize}
\begin{split}
    |\mathcal{C}_{\mathrm{pool}}| &\ge \sum_{v \in V} N_{v_{\mathrm{root}}, v}^{(n')}(P) \\
    &\ge 2^{n' H(P) + \frac{|V| - |E|}{2} \log_2(2\pi n') + \Delta(n')} \sum_{v \in V} 2^{C_v} .
\end{split}
\end{equation}
The resulting rate of our construction, $R_{\mathrm{pool}}$, is 
\begin{align*}
    R_{\mathrm{pool}} = \frac{1}{n} \log_2 |\mathcal{C}_{\mathrm{pool}}| 
    & \ge \frac{n'}{n} H(P) + \frac{|V| - |E|}{2n} \log_2(2\pi n') \\
    &\ \ \ \  + \frac{\Delta(n')}{n}+\frac{1}{n} \log_2 \sum_{v \in V} 2^{C_v}.
\end{align*}

Now, $n'/n = \alpha$, and as we let $z \to \infty$, we also have $n  \to \infty$, so that the last three terms in the lower bound above all go to $0$. Then, letting $\alpha \to 1$, we have
\begin{equation}
    \lim_{\alpha \to 1} \lim_{n \to \infty} R_{\mathrm{pool}} \ge H(P) .
\end{equation}
Finally, recall that $P = P^{(n_0)}$ was chosen to be such that ${\|P-\hat{P}\|}_{\infty} < \xi$. Letting $\xi \to 0$, so that $P \to \hat{P}$, we see that $R_{\mathrm{pool}}$ can be made arbitrarily close to $H(\hat{P}) = \mathrm{cap}(S_{\hat{\mathbf{r}},\epsilon}(G))$. We conclude that our proposed construction of weakly constrained codebooks via Eulerian cycles is asymptotically capacity-achieving.

\section{Error-Correcting Weakly Constrained Codes}
\label{sec:error_correcting_wcc}
While the construction of $\mathcal{C}_{\mathrm{pool}}$ described in Section~\ref{sec:weak_codes_randomwalks} is primarily designed to achieve capacity, the codewords inherently exhibit good Hamming distance. We leverage this property by applying an \emph{expurgation} procedure to extract an error-correcting code $\mathcal{C}_{\mathrm{ec}} \subseteq \mathcal{C}_{\mathrm{pool}}$. Specifically, for a target relative distance $\delta$, we want to extract a subset of codewords in which every distinct pair maintains a relative Hamming distance at least $\delta$.

By identifying and discarding \emph{bad pairs}, i.e., those whose relative distance falls below the target threshold, we aim to retain a subset large enough to preserve a strictly positive asymptotic rate while gaining error-correction capability. This approach transforms the challenge of codebook construction into an analysis of \emph{distance concentration}: if the Hamming distance is tightly concentrated around its mean and the target $\delta$ is chosen below this mean, the fraction of bad pairs remains small. This ensures that the expurgated codebook $\mathcal{C}_{\mathrm{ec}}$ retains a large cardinality.

\subsection{Distance Properties of $\mathcal{C}_{\mathrm{pool}}$ and Expurgation}
Since the Hamming distance is non-negative and additive, the distance between the prefixes, $d_H(\mathbf{w}, \mathbf{w}')$, serves as a lower bound for the total distance between the codewords, i.e., $d_H(\mathbf{c}, \mathbf{c}') \ge d_H(\mathbf{w}, \mathbf{w}')$. Therefore, ensuring a minimum distance between the admissible prefixes is sufficient to guarantee the error-correction capability of the entire codebook. Therefore, we restrict our attention to distance concentration in these length-$n'$ prefixes. We now formally define a bad pair. 
 
\begin{definition}[Bad Pair] \label{def:bad_pair}
For a target relative distance $\delta$, a pair of codewords $(\mathbf{c}, \mathbf{c}')$ is defined as a \emph{bad pair} if their corresponding prefixes satisfy 
\[
d_H(\mathbf{w}, \mathbf{w}')/n' < \delta.
\]
\end{definition}

Using this definition, we construct a graph $\mathcal{G}_{\mathrm{pool}} = (\mathcal{C}_{\mathrm{pool}}, E_{\mathrm{bad}})$, where the vertex set is the pool of codewords and an edge exists between distinct vertices $\mathbf{c}, \mathbf{c}' \in \mathcal{C}_{\mathrm{pool}}$ if they form a bad pair. An \emph{independent set} in $\mathcal{G}_{\mathrm{pool}}$ corresponds to an error-correcting code $\mathcal{C}_{\mathrm{ec}}$ with relative distance at least $\delta$. To quantify the size of this set, we bound the probability mass associated with the edges 
\begin{equation}
\label{eq: E_bad}
    P(E_{\mathrm{bad}}) = \sum_{(\mathbf{c}, \mathbf{c}') \in E_{\mathrm{bad}}} P(\mathbf{c})P(\mathbf{c}'),
\end{equation}
where $P(\mathbf{c})$ is, as defined before, the probability of the prefix $\mathbf{w}$ being generated by a random walk of length $n'$ on $G$ governed by the transition probabilities of $P$. By bounding $ P(E_{\mathrm{bad}}) $, we can apply the probabilistic version of Turán’s theorem~\cite{Turan} to obtain a lower bound on the size of the independent set in $\mathcal{G}_{\mathrm{pool}}$, as subsequently described in the proof of Theorem \ref{thm:achievable_rate}.

In order to bound $P(E_{\mathrm{bad}})$, we must first characterize the expected distance of $\mathcal{C}_{\mathrm{pool}}$, which serves as the reference point for our target relative distance  $\delta$.

\begin{lemma}[Expected Hamming Distance]
\label{lem:expected_distance}
Let $\mathbf{w}$ and $\mathbf{w}'$ be length-$n'$ prefixes that evolve independently and are initialized according to the stationary distribution $\pi$, the expected Hamming distance satisfies
\begin{equation}
  \mathbb{E}[d_H(\mathbf{w}, \mathbf{w'})] = n'(1 - S),
\end{equation}
where $S = \sum_{a \in \Sigma} \left( \sum_{e: L(e)=a} P(e) \right)^2$ is the label collision probability.
\end{lemma}

\begin{IEEEproof}
By the linearity of expectation, the expected distance is the sum of mismatch probabilities at each step, $\mathbb{E}[d_H(\mathbf{w}, \mathbf{w'})] = \sum_{k=1}^{n'} P(L(e_k) \neq L(e'_k))$. Since the walks are initialized according to the stationary distribution, the probability of traversing a specific edge $e$ at any step $k$ is exactly its stationary probability $P(e)$. Consequently, the probability that both independent random walks emit the same label $a \in \Sigma$ at time $k$ is
\begin{equation*}
    P(L(e_k) = a \text{ and } L(e'_k) = a) = \left( \sum_{e: L(e)=a} P(e) \right)^2 \text{.}
\end{equation*}
Summing over all symbols $a \in \Sigma$ yields the label collision probability $S$. Thus, the mismatch probability at each step is $1 - S$, and the total expected distance over $n'$ steps is $n'(1 - S)$.
\end{IEEEproof}

We choose the target $\delta$ to be strictly less than the expected relative distance $(1-S)$. Specifically, we define a back-off parameter $\epsilon > 0$ such that 
\begin{equation}
    \delta = (1-S) - \epsilon.
\end{equation}
Equivalently, the target absolute Hamming distance between length-$n'$ prefixes $\mathbf{w}$ and $\mathbf{w}'$ that evolve independently and are initialized according to the stationary distribution $\pi$ is
\begin{equation}
    d_H(\mathbf{w}, \mathbf{w}') = n'(1-S) - n'\epsilon.
\end{equation}

Now, we formally define a failure probability $P_{\mathrm{fail}}(\epsilon)$ which is going to be a crucial parameter for our subsequent distance concentration analysis.

\begin{definition}[Failure Probability] \label{def:p_fail}
For a given distance back-off parameter $\epsilon > 0$, the failure probability $P_{\mathrm{fail}}(\epsilon)$ denotes the probability that two independent random walks sampled from the Markov chain $P$ on $G$, form a bad pair i.e., 
\begin{equation}
    P_{\mathrm{fail}}(\epsilon) \triangleq P\left( d_H(\mathbf{w}, \mathbf{w}') \leq n'(1-S-\epsilon) \right).
\end{equation}
\end{definition}

\subsection{Distance Concentration Analysis}

To analyze the concentration of the Hamming distance between two independent random walks generating the prefixes $\mathbf{w}$ and $\mathbf{w'}$, we model them as a random walk $\mathbf{Y}$ on a product graph $\mathcal{G} = G \times G$. The product graph is defined such that its vertex set is $V \times V$, where a directed edge $\mathbf{e}\in E_{\mathcal{G}}$ exists from $(u, v)$ to $(u', v')$ if and only if $e: u \to u'$ and $e': v \to v'$ are valid edges in the primal graph $G$. Thus, each edge $\mathbf{e} \in E_{\mathcal{G}}$ corresponds to the pair $(e, e')$ and inherits the labels $(L(e), L(e'))$. We define a mismatch indicator function $f: E_{\mathcal{G}} \to \{0, 1\}$, where $f(e, e') = 1$ if $L(e) \neq L(e')$ and $f(e, e') = 0$ otherwise.

This joint walk is governed by the product Markov chain $P' = P \times P$, with transition probabilities given by
\begin{equation}
    P'((u', v')| (u, v)) = p(u'| u) p(v'| v),
\end{equation}
for all $(u, u'), (v, v') \in E$. Since $P$ is a stationary Markov chain on $G$, the product chain $P'$ is also stationary on $\mathcal{G}$ (see, e.g., Exercise~12.6 in~\cite{staionaryexample}). 

We now analyze the concentration of distance in the finite and asymptotic regimes.

\paragraph{Finite Blocklength Analysis}
\label{subsec:finite_length}

We use Bernstein inequalities to analyze the distance concentration for finite blocklengths. Let $\mathbf{Y} = (\mathbf{e}_1, \mathbf{e}_2, \dots, \mathbf{e}_{n'})$ be a random walk on the product graph $\mathcal{G}$ governed by the Markov chain $P'$, where each product edge $\mathbf{e}_k$ corresponds to the pair $(e_k, e'_k)$. The Hamming distance $d_H(\mathbf{w}, \mathbf{w}')$ can then be expressed as
\begin{equation}
\label{eq: distance_sum}
    d_H(\mathbf{w}, \mathbf{w}') = \sum_{k=1}^{n'} f(e_k, e'_k).
\end{equation}
Under the stationary initialization of the walk $\mathbf{Y}$, the expectation and variance of $f$ are given by $\mathbb{E}_{P'}[f] = 1 - S$ and $\text{Var}_{P'}(f) = S(1 - S)$, respectively.

\begin{lemma}[Finite-Length Bound on $P_{\text{fail}}$]
\label{lem:finite_pfail}
Let $\mathbf{w}$ and $\mathbf{w}'$ be the length-$n'$ prefixes of the codewords $\mathbf{c}, \mathbf{c}' \in \mathcal{C}_{\mathrm{pool}}$, both starting from the fixed vertex $v_{\text{root}}$. For a non-stationarity constant $N_q = 1/\pi'((v_{\text{root}},v_{\text{root}}))$ where $\pi'$ is the stationary distribution of the product chain $P' = P \times P$ and any $\epsilon > 0$, the failure probability $P_{\mathrm{fail}}(\epsilon)$ is bounded as follows:
\begin{enumerate}
    \item If the product chain $P'$ is non-reversible, then
    \begin{equation*}
       P_{\mathrm{fail}}(\epsilon) \le \sqrt{2N_q} \exp \left( -n' \frac{\epsilon^2 \gamma_{\mathrm{ps}}}{16S(1-S)(1+\frac{1}{\gamma_{\mathrm{ps}}}) + 40\epsilon S} \right),
    \end{equation*}
    where $\gamma_{\mathrm{ps}}$ is the pseudo-spectral gap of $P'$. 
    
    \item If the product chain $P'$ is reversible, then
    \begin{equation*}
       P_{\mathrm{fail}}(\epsilon) \le \sqrt{2N_q} \exp \left( -n' \frac{\epsilon^2 \gamma}{8S(1-S) + 20\epsilon S} \right),
    \end{equation*}
    where $\gamma$ is the spectral gap of $P'$.
\end{enumerate}
\end{lemma}

\begin{IEEEproof}
We have expressed the Hamming distance as a sum of dependent random variables in \eqref{eq: distance_sum}. We use the Bernstein concentration inequalities to bound $P_{\mathrm{fail}}(\epsilon)$. Initially, we assume that the joint walk $\mathbf{Y}$ is initialized according to the stationary distribution of $P'$, i.e., $\mathbf{e}_1 \sim \pi \times \pi$.

For notational brevity, let $d_H$ denote the Hamming distance $d_H(\mathbf{w}, \mathbf{w}')$. For non-reversible Markov chains, applying the Bernstein inequality in \eqref{eq:ber_non-rev}, we get 
\begin{equation*}
P_{\pi}(|d_H - \mathbb{E}[d_H]| \ge t) \le 2 \exp \left( -\frac{t^2 \gamma_{\mathrm{ps}}}{8(n' + \gamma_{\mathrm{ps}}^{-1}) V_f + 20tS} \right),
\end{equation*}
where $V_f = S(1-S)$ is the variance of the mismatch indicator function. Setting the deviation $t = n'\epsilon$ and noting that $(1+\frac{1}{n' \gamma_{\mathrm{ps}}}) < (1+\frac{1}{\gamma_{\mathrm{ps}}})$, the probability of $(\mathbf{w}, \mathbf{w}')$ forming a bad pair simplifies to
\begin{equation*}
    P_{\pi}(d_H \le n'(1-S - \epsilon)) \le 2 \exp \left( -\frac{n' \epsilon^2 \gamma_{\mathrm{ps}}}{8 V_f (1+\frac{1}{\gamma_{\mathrm{ps}}}) + 20 \epsilon S} \right).
\end{equation*}

When the product chain $P'$ is reversible, we apply the corresponding Bernstein inequality from \eqref{eq:ber_rev}. For a deviation $t$, the concentration of the Hamming distance $d_H$ is bounded as
\begin{equation*}
    P_{\pi}(|d_H - \mathbb{E}[d_H]| \ge t) \le 2 \exp \left( -\frac{t^2 \gamma}{4n'V_f + 10tS} \right).
\end{equation*}
Substituting $t = n'\epsilon$, we get
\begin{equation*}
    P_{\pi}(d_H \le n'(1-S - \epsilon)) \le 2 \exp \left( -\frac{n' \epsilon^2 \gamma}{4 V_f + 10 \epsilon S} \right).
\end{equation*}

To account for the fixed initialization at $v_{\text{root}}$, we apply the non-stationary bound from \eqref{eq:ber_nonsta}. This introduces the pre-factor $\sqrt{2N_q}$ and results in a factor of $2$ loss in the exponent, yielding the claimed bound.
\end{IEEEproof}

\paragraph{Asymptotic analysis}
\label{subsec:asymptotic_analysis}
We use the LDP for Markov chains (Section~\ref{subsec:ldp_prelim}) to analyse the distance concentration in the asymptotic regime. For a random walk $\mathbf{Y} = ((e_1, e'_1), \dots, (e_{n'}, e'_{n'}))$ on the product graph $\mathcal{G}$, we define its pair empirical measure $P_{\mathbf{w}\mathbf{w}'} \in \mathcal{M}_s(G \times G)$ as
\begin{equation}
    P_{\mathbf{ww'}}(e, e') = \frac{1}{n'} \sum_{k=1}^{n'} \mathbb{I}((e_k, e'_k) = (e, e')) \text{.}
\end{equation}
To ensure these trajectories correspond to valid codewords in $\mathcal{C}_{\mathrm{pool}}$, we characterize the marginal distributions of $P_{\mathbf{ww'}}$ on the individual walks $\mathbf{w}$ and $\mathbf{w'}$ as
\begin{align}
    P_\mathbf{w}(e) &= \sum_{e' \in E}  P_{\mathbf{ww'}}(e, e'), \\
    P_\mathbf{w'}(e') &= \sum_{e \in E}  P_{\mathbf{ww'}}(e, e').
\end{align}

$P(E_{\mathrm{bad}})$ defined in \eqref{eq: E_bad}, is equivalent to the probability of the joint empirical measure $P_{\mathbf{ww'}}$ belonging to the set $A_{\zeta, \epsilon} \subset \mathcal{M}_s(G \times G)$, which is a set of \emph{valid-bad} joint distributions. This set is characterized by two primary conditions:

\begin{itemize}
    \item Typicality: The marginals of the joint empirical measure must be $\zeta$-typical with respect to the Markov chain $P$.
    \item Distance Failure: The empirical relative Hamming distance must be at most the target, $\delta = 1-S-\epsilon$.
\end{itemize}

Formally, we define this set $A_{\zeta, \delta}$ as
{
\small
\begin{equation}
\label{eq:SetA}
    A_{\zeta, \delta} = \left\{ Q \in \mathcal{M}_s(G \times G) : 
    \begin{aligned}
        &|Q_w(e) - P(e)| < \zeta, \quad \forall e \in E, \\
        &|Q_{w'}(e) - P(e)| < \zeta, \quad \forall e \in E, \\
        &\mathbb{E}_Q[f] \le \delta
    \end{aligned}
    \right\} \text{,}    
\end{equation}
}

\begin{lemma}[Asymptotic Bound on $P(E_{\mathrm{bad}})$]
\label{lem:asymptotic_pfail}
Let $\mathbf{w}$ and $\mathbf{w}'$ be the length-$n'$ prefixes of codewords $\mathbf{c}, \mathbf{c}' \in \mathcal{C}_{\mathrm{pool}}$. For any typicality parameter  $0 < \zeta < (\frac{1-\alpha}{\alpha})P_{\min}$, where $P_{\min} = \min_{e \in E} P(e)$ and distance back-off $\epsilon > 0$, the failure probability $P(E_{\mathrm{bad}})$ for target $\delta=1-S-\epsilon$ satisfies
\begin{equation*}
    \limsup_{n' \to \infty} \frac{1}{n'} \log P(E_{\mathrm{bad}})\le -\inf_{Q \in A_{\zeta, \delta}} \Big( 2H(P) - H(Q)\Big) -\mathcal{O}(\zeta) \text{,}
\end{equation*}
where $A_{\zeta, \delta}$ is the set of stationary joint measures defined in \eqref{eq:SetA} and the constant in the $\mathcal{O}$ notation depends only on the Markov chain $P$.
\end{lemma}

\begin{IEEEproof}
We use the LDP for the pair empirical measure of a Markov chain to bound the asymptotic $P(E_{\mathrm{bad}})$. The sequence of measures $\mu_{n'}(\cdot) = P_{\pi}(P_{\mathbf{ww'}} \in \cdot)$ satisfies an LDP with rate $n'$ and rate function $I(Q)$. For any stationary Markov chain $Q$ on the product graph $\mathcal{G}$, the rate function is
\begin{equation*}
    I(Q) = \sum_{\mathbf{e}=(e,e') \in \mathcal{G}} Q(\mathbf{e}) \log \frac{Q(\tau(\mathbf{e}) | \sigma(\mathbf{e}))}{P'(\tau(\mathbf{e}) | \sigma(\mathbf{e}))} \text{.}
\end{equation*}
Expanding the logarithm and using the fact that transitions on the product graph factorize as $P'(\tau(\mathbf{e}) | \sigma(\mathbf{e})) = p(\tau(e) | \sigma(e)) p(\tau(e') | \sigma(e'))$ for $\mathbf{e} = (e, e')$, we obtain
\begin{align*}
    I(Q) &= -H(Q) - \sum_{(e, e') \in \mathcal{G}} Q(e, e') \log p(\tau(e) | \sigma(e)) \nonumber \\
    &\quad - \sum_{(e, e') \in \mathcal{G}} Q(e, e') \log p(\tau(e') | \sigma(e')) \text{.}
\end{align*}
By marginalizing the second term over $Q_1(e) = \sum_{e'} Q(e, e')$ and $Q_2(e') = \sum_{e} Q(e, e')$, we decompose the summation as
\begin{multline*}
    I(Q) = -H(Q) - \sum_{e \in G} Q_1(e) \log p(\tau(e) | \sigma(e)) \\ - \sum_{e' \in G} Q_2(e') \log p(\tau(e') | \sigma(e')) \text{.}
\end{multline*}
    
For any $Q \in A_{\zeta, \delta}$, the marginals satisfy the typicality constraints, $|Q_1(e) - P(e)| < \zeta$ and $|Q_2(e) - P(e)| < \zeta$. Therefore, substituting $Q_1(e) = P(e) - \zeta$, we get 
\begin{equation}
    -\sum_{e \in G} Q_1(e) \log p(\tau(e) | \sigma(e)) \geq H(P) + \sum_{e \in G} \zeta \log p(\tau(e) | \sigma(e)) \text{.}
\end{equation}
Applying this to both marginal terms, we get $I(Q) \ge 2H(P) - H(Q) + \mathcal{O}(\zeta)$, where the constant in the $\mathcal{O}$ notation depends only on the Markov chain $P$. The result follows from the LDP upper bound applied to the set $A_{\zeta, \delta}$.
\end{IEEEproof}

\subsection{Main Result: Achievable Rate}

We now bound the size of an independent set in $\mathcal{G}_{\mathrm{pool}}$, which corresponds to the size of the expurgated codebook $|\mathcal{C}_{\mathrm{ec}}|$, using the bounds on $P_{\mathrm{fail}}(\epsilon)$ and $P(E_{\mathrm{bad}})$ established in Lemmas~\ref{lem:finite_pfail} and~\ref{lem:asymptotic_pfail}, respectively. We first establish the following two claims essential for characterizing the size of the expurgated codebook $\mathcal{C}_{\mathrm{ec}}$.

\begin{claim}[Codeword Probability Ratio]
\label{claim:prob_ratio}
For any two codewords $\mathbf{u}, \mathbf{v} \in \mathcal{C}_{\mathrm{pool}}$, the ratio of their generation probabilities is bounded by
\begin{equation*}
    \frac{P(\mathbf{c})}{P(\mathbf{c'})} \le 2^{-2n\alpha \zeta \sum_{e \in E} \log p(\tau(e)\mid \sigma(e))} \text{.}
\end{equation*}

\end{claim}

\begin{IEEEproof}
The generation probability of a walk $\mathbf{w}$ is given by $P(\mathbf{w}) = \prod_{e \in E} p(\tau(e) \mid \sigma(e))^{S_{\mathbf{w}}(e)}$, where $S_{\mathbf{w}} (e)$ denotes the count of edge $e$ in the walk. For any $\mathbf{c}, \mathbf{c'} \in \mathcal{C}_{\mathrm{pool}}$, the ratio of their generation probabilities satisfies
\begin{equation*}
    \log \frac{P(\mathbf{c})}{P(\mathbf{c'})} = \sum_{e \in E} (S_\mathbf{c}(e) - S_\mathbf{c'}(e)) \log p(\tau(e) \mid \sigma(e)) \text{.}
\end{equation*}
By the definition of $\mathcal{C}_{\mathrm{pool}}$, the counts satisfy $n'(P(e) - \zeta) < S(e) < n'(P(e) + \zeta)$ for all $e \in E$. To obtain the upper bound, we set $S_\mathbf{c}(e) = n'(P(e) - \zeta)$ and $S_\mathbf{c'}(e) = n'(P(e) + \zeta)$ for all $e \in E$. Substituting $n' = n\alpha$, we obtain
\begin{equation*}
    \log \frac{P(\mathbf{c})}{P(\mathbf{c'})} \le -2n\alpha\zeta \sum_{e \in E}  \log p(\tau(e) \mid \sigma(e)) \text{.}
\end{equation*}
Exponentiating both sides establishes the result.
\end{IEEEproof}

\begin{claim}[Bound on Maximum Generation Probability]
\label{claim:p_v_bound}
For any codeword $\mathbf{c} \in \mathcal{C}_{\mathrm{pool}}$, the maximum generation probability $\bar{P}_{\max}$ is bounded as
\begin{equation*}
    \bar{P}_{\max} \le \frac{2^{-2\alpha n \zeta \sum_{e \in E} \log p(\tau(e) \mid \sigma(e))}}{|\mathcal{C}_{\mathrm{pool}}|} \text{.}
\end{equation*}
\end{claim}

\begin{IEEEproof}
The sum of generation probabilities over the set of codewords in $\mathcal{C}_{\mathrm{pool}}$ satisfies $\sum_{\mathbf{c} \in \mathcal{C}_{\mathrm{pool}}} P(\mathbf{c}) \le 1$. By replacing each term in the summation with the minimum codeword probability, we obtain the following bound
\begin{equation}
\label{eq:cl2-proof}
    \min_{\mathbf{c} \in \mathcal{C}_{\mathrm{pool}}} P(\mathbf{c}) \le \frac{1}{|\mathcal{C}_{\mathrm{pool}}|} \text{.}
\end{equation}

From the bound established in Claim~\ref{claim:prob_ratio}, we know that the maximum probability $\bar{P}_{\max}$ satisfies 
\[
\frac{\bar{P}_{\max}}{\min_{\mathbf{c} \in \mathcal{C}_{\mathrm{pool}}} P(\mathbf{c})} \le 2^{-2\alpha n \zeta \sum_{e \in E} \log p(\tau(e) \mid \sigma(e))}.
\]
Substituting the bound for the minimum probability in \eqref{eq:cl2-proof} completes the proof.
\end{IEEEproof}

We now characterize the achievable rate of the expurgated codebook $\mathcal{C}_{\mathrm{ec}}$.
\begin{theorem}[Achievable Rate for Error-Correcting Weakly Constrained Codes]
\label{thm:achievable_rate}
There exists a weakly constrained code $\mathcal{C}_{\mathrm{ec}} \subseteq \mathcal{C}_{\mathrm{pool}}$ with relative distance $\delta = \alpha(1-S-\epsilon)$ 
such that
\begin{equation*}
    |\mathcal{C}_{\mathrm{ec}}| \ge \frac{P(\mathcal{C}_{\mathrm{pool}})^{2}}{4 P'_{\mathrm{fail}}(\epsilon)} \text{,}
\end{equation*}
where $P'_{\mathrm{fail}}(\epsilon)$ is the effective failure probability defined as
\begin{equation*}
 P'_{\mathrm{fail}}(\epsilon) \triangleq \max\left( P_{\mathrm{fail}}(\epsilon), \frac{\bar{P}_{\max}P(\mathcal{C}_{\mathrm{pool}})}{2} \right).
\end{equation*}
Here, $P_{\mathrm{fail}}(\epsilon)$ is the failure probability from Lemma~\ref{lem:finite_pfail}, $P(\mathcal{C}_{\mathrm{pool}})$ is the total probability mass of the typical pool from Theorem~\ref{thm:reliability}, and $\bar{P}_{\max}$ is bounded in Claim \ref{claim:p_v_bound}.
In the asymptotic regime the rate $R_{\mathrm{ec}}(\delta)$ is
\begin{equation*}
  R_{\mathrm{ec}}(\delta)\geq \min(R_1,R_2), 
\end{equation*}
where
 \begin{equation*}
  R_1 \ge \alpha\Big(2H(P)-\sup_{Q \in A_{\zeta,\delta}} H(Q)+\mathcal{O}(\zeta)\Big)
 \end{equation*}
$A_{\zeta, \delta}$ being the set defined in \eqref{eq:SetA}, and
\begin{equation*}
        R_2 \ge \alpha H(P) - \mathcal{O}(\zeta)\text{.}
\end{equation*}

The constant in the $\mathcal{O}$-notation depends only on the Markov chain $P$.
\end{theorem}

\begin{IEEEproof}
The proof follows a randomized algorithm for construction of an independent set, similar to that in \cite{Turan}; the algorithm is described as follows:
\begin{itemize}
    \item Vertex Deletion: Each vertex $v \in \mathcal{G}_{\mathrm{pool}}$ is deleted independently with probability $1- P_{v}$, where $P_{v} = z |\mathcal{C}_{\mathrm{pool}}|P(v)$. $P(v)$ denotes the probability of generating the random walk corresponding to vertex $v$ and $z$ is a parameter to be optimized later.
    \item Bad pair removal:  For every edge $(u,v)\in E_{\mathrm{bad}}$ whose both endpoints survive the vertex deletion step, one endpoint is removed to eliminate the bad pair.
\end{itemize}

Let $N_\mathrm{v}$ and $N_\mathrm{e}$ denote the numbers of vertices and edges, respectively, that survive the first step. Since $N_\mathrm{e}$ represents the remaining bad pairs, an independent set $\mathcal{C}_{\mathrm{ec}}$ can be obtained by removing at most one vertex from each surviving edge. Therefore, we get
\begin{equation}
 \label{eq:expectedsize}
\mathbb{E}[|\mathcal{C}_{\mathrm{ec}}|] \geq \mathbb{E}[N_\mathrm{v}]-\mathbb{E}[N_\mathrm{e}].   
\end{equation}
Further,
\[
\mathbb{E}[N_\mathrm{v}]=z |\mathcal{C}_{\mathrm{pool}}|\sum_{v\in \mathcal{C}_{\mathrm{pool}}}P(v) \geq z |\mathcal{C}_{\mathrm{pool}}|P(\mathcal{C}_{\mathrm{pool}}), 
\]
where $P(\mathcal{C}_{\mathrm{pool}})$ is defined in Theorem \ref{thm:reliability} and 
\[
\mathbb{E}[N_\mathrm{e}]\leq z^{2}|\mathcal{C}_{\mathrm{pool}}|^{2}\sum_{(u,v)\in E_{\mathrm{bad}}} P(u)P(v) =z^{2}|\mathcal{C}_{\mathrm{pool}}|^{2}P(E_{\mathrm{bad}}).
\]

\paragraph{Finite blocklengths}
While $P(E_{\mathrm{bad}})$ is the probability of a pair of walks generating codewords in $\mathcal{C}_{\mathrm{pool}}$ that form a bad pair, $P_{\mathrm{fail}}(\epsilon)$ is the probability of any two walks forming a bad pair regardless of their membership in $\mathcal{C}_{\mathrm{pool}}$. Therefore, we have $P(E_{\mathrm{bad}}) \le P_{\mathrm{fail}}(\epsilon)$, which gives
\[
\mathbb{E}[N_\mathrm{e}]\leq z^{2}|\mathcal{C}_{\mathrm{pool}}|^{2}P_{\mathrm{fail}}(\epsilon) \leq z^{2}|\mathcal{C}_{\mathrm{pool}}|^{2}P'_{\mathrm{fail}}(\epsilon) \ .
\]
where $P'_{\mathrm{fail}}(\epsilon)= \max\left( P_{\mathrm{fail}}(\epsilon), \frac{\bar{P}_{\max} P(\mathcal{C}_{\mathrm{pool}})}{2}\right)$ where $P_{\mathrm{fail}}(\epsilon)$ is bounded in   Lemma~\ref{lem:finite_pfail} and $\bar{P}_{\max}$ is bounded in Claim \ref{claim:p_v_bound} .

Substituting in \eqref{eq:expectedsize}, we get
\[
\mathbb{E}[|\mathcal{C}_{\mathrm{ec}}|] \geq z |\mathcal{C}_{\mathrm{pool}}|P(\mathcal{C}_{\mathrm{pool}})-z^2 |\mathcal{C}_{\mathrm{pool}}|^{2}P'_{\mathrm{fail}}(\epsilon)
\]
Optimizing the codebook size, by choosing $z=\frac{P(\mathcal{C}_{\mathrm{pool}})}{2P'_{\mathrm{fail}}(\epsilon)|\mathcal{C}_{\mathrm{pool}}|}$, yields
\[
\mathbb{E}[|\mathcal{C}_{\mathrm{ec}}|] \geq \frac{P(\mathcal{C}_{\mathrm{pool}})^{2}}{4P'_{\mathrm{fail}}(\epsilon)}.
\]
The parameter $z$ must be chosen such that the selection probability satisfies $P_{v} \le 1$ for all $v \in \mathcal{C}_{\mathrm{pool}}$. This requirement imposes a validity condition, expressed as
\begin{equation*}
    \forall v \in \mathcal{C}_{\mathrm{pool}}, \quad P(v) \leq \frac{2 P'_{\mathrm{fail}}(\epsilon)}{ P(\mathcal{C}_{\mathrm{pool}})}.
\end{equation*}
This condition is satisfied for the chosen $P'_{\mathrm{fail}}(\epsilon)$ as $P'_{\mathrm{fail}}(\epsilon) \geq \frac{\bar{P}_{\max} P(\mathcal{C}_{\mathrm{pool}})}{2}$.
\paragraph{Asymptotic blocklengths}
For asymptotic blocklengths, we have
\[
\mathbb{E}[N_\mathrm{e}]\leq z^{2}|\mathcal{C}_{\mathrm{pool}}|^{2}P(E_{\mathrm{bad}}) \leq z^{2}|\mathcal{C}_{\mathrm{pool}}|^{2}P'(E_{\mathrm{bad}}) \ .
\]
where $P(E_{\mathrm{bad}})$ is as bounded in Lemma~ \ref{lem:asymptotic_pfail} and $P'(E_{\mathrm{bad}})\triangleq \max\left( P(E_{\mathrm{bad}}),\frac{\bar{P}_{\max}P(\mathcal{C}_{\mathrm{pool}})}{2}\right)$. Proceeding as in the finite blocklength case, we obtain
\[
\mathbb{E}[|\mathcal{C}_{\mathrm{ec}}|]\geq \frac{P(\mathcal{C}_{\mathrm{pool}})^{2}}{4P'(E_{\mathrm{bad}})}
\]
Therefore, there exists an expurgated codebook satisfying
\[
|\mathcal{C}_{\mathrm{ec}}|
\ge
\mathbb{E}[|\mathcal{C}_{\mathrm{ec}}|].
\]
The asymptotic rate $R_{\mathrm{ec}}(\delta) =\lim_{n \to \infty} \frac{1}{n} \log_2 |\mathcal{C}_{\mathrm{ec}}|$ depends on which term dominates in $P'(E_{\mathrm{bad}})$,
\begin{itemize}
    \item Case 1: When $P'(E_{\mathrm{bad}}) = P(E_{\mathrm{bad}})$, applying Lemma~\ref{lem:asymptotic_pfail} yields the rate
\begin{align*}
R_1&\ge\alpha \Big(\inf_{Q \in A_{\zeta,\delta}}\Big(2H(P)-H(Q)\Big)+\mathcal{O}(\zeta)\Big) \\
&=\alpha \bigl(2H(P)-\sup_{Q \in A_{\zeta,\delta}} H(Q)+\mathcal{O}(\zeta)\bigr).
\end{align*}
    \item Case 2: When $P'(E_{\mathrm{bad}}) = \frac{\bar{P}_{\max}P(\mathcal{C}_{\mathrm{pool}})}{2}$, applying Claim~\ref{claim:p_v_bound} yields the rate
    \begin{align*}
    R_2 & \ge \alpha ( H(P) + 2\zeta \sum_{e \in E} \log_2 p(\tau(e) \mid \sigma(e)) )\\
    & =
    \alpha H(P) -\mathcal{O}(\zeta).
    \end{align*}
\end{itemize}

Therefore, the asymptotic rate $R_{\mathrm{ec}}(\delta) \geq \min(R_1,R_2).$
\end{IEEEproof}
\begin{remark}[Connection to the Gilbert-Varshamov Bound]
As we let $\alpha \to 1$, the condition $0 < \zeta < (\frac{1-\alpha}{\alpha}) P_{\min}$ forces $\zeta \to 0$. Moreover, $\delta \to 1-S-\epsilon$. Hence, as $\alpha \to 1$, the lower bound on $R_1$ converges to $2H(P) - \sup_{Q \in A_{0, 1-S-\epsilon}} H(Q)$. This is the same as the improved Gilbert-Varshamov (GV) bound for constrained systems established in \cite{lowerupper1}.
\end{remark}

\section{Efficient Construction via Concatenation}
\label{sec:concatenation}

Our construction in Section~\ref{sec:error_correcting_wcc} establishes the \emph{existence} of error-correcting weakly constrained codes but does not yield an efficient encoding or decoding scheme. We therefore propose a concatenated construction that preserves the distance guarantees of the weakly constrained codes while enabling efficient encoding and decoding.

\paragraph{Inner Code ($\mathcal{C}_{\mathrm{in}}$)}

We use a subset of the expurgated error-correcting weakly constrained code $\mathcal{C}_{\mathrm{in}} \subseteq \mathcal{C}_{\mathrm{ec}}$ of size $q$ as the inner code, where $q$ is largest prime power such that $q \leq |\mathcal{C}_{\mathrm{ec}}|$. The inner codebook is given by $\mathcal{C}_{\mathrm{in}} =\{\mathbf{c}_1,\mathbf{c}_2,\dots,\mathbf{c}_{q}\}$, where each codeword $\mathbf{c}_i\in \mathcal{C}_{\mathrm{ec}}$. The inner code has blocklength $n$ and minimum Hamming distance
\[
d_{\mathrm{in}} = n\alpha(1-S-\epsilon)\text{.}
\]
Bertrand's Postulate states that, for any integer $M > 1$, there exists a prime $p$ such that $M < p < 2M$. Therefore, $q > |\mathcal{C}_{\mathrm{ec}}|/2$. Let $R_{\mathrm{in}} = \frac{1}{n} \log_2 |\mathcal{C}_{\mathrm{ec}}|$ denote the rate of the expurgated code derived in Theorem~\ref{thm:achievable_rate}. We define the rate of our inner code as $R'_{\mathrm{in}} = \frac{1}{n} \log_2 q$, then
\begin{equation*}
R'_{\mathrm{in}} > \frac{\log_2(|\mathcal{C}_{\text{ec}}|/2)}{n} = \frac{\log_2 |\mathcal{C}_{\text{ec}}| - 1}{n} = R_{\mathrm{in}} - \frac{1}{n} \text{.}
\end{equation*}
 

\paragraph{Outer Code ($\mathcal{C}_{\mathrm{out}}$)}
As the outer code, we employ a Reed-Solomon (RS) code over an alphabet $\Sigma_{\mathrm{out}}$ of size $q$. The outer code has blocklength $N_{\mathrm{out}} = q - 1$, dimension $K$, and minimum distance
\[
D_{\mathrm{out}} = N_{\mathrm{out}} - K + 1.
\]

\paragraph{Encoding}
A message is first encoded by the outer RS encoder into a sequence $(s_1,s_2,\dots,s_{N_{\mathrm{out}}})$, where $s_i \in \Sigma_{\mathrm{out}}$. Each symbol $s_i$ is then mapped to the corresponding inner codeword $\mathbf{c}_{s_i}\in\mathcal{C}_{\mathrm{in}}$. The final codeword is obtained by concatenation
\[
\mathbf{C}
= \mathbf{c}_{s_1} \| \mathbf{c}_{s_2} \| \cdots \| \mathbf{c}_{s_{N_{\mathrm{out}}}}.
\]

Any $\mathbf{c}_{s_{i}}\in \mathcal{C}_{\mathrm{in}}$ is an Eulerian cycle in $G_n$ rooted at a fixed vertex $v_\mathrm{root} \in G$. Therefore, for any $\mathbf{c}_i, \mathbf{c}_j \in \mathcal{C}_{\mathrm{in}}$, $\mathbf{c}_i$ terminates at $v_\mathrm{root}$ and $\mathbf{c}_j$ originates at $v_\mathrm{root}$, so their concatenation $\mathbf{c}_i\|\mathbf{c}_j$ forms a valid walk on $G$. Since each cycle individually satisfies the Eulerian edge count constraints, their concatenation preserves these constraints globally. Consequently, the concatenated sequence $\mathbf{c}_i\|\mathbf{c}_j$ is an Eulerian cycle on $G_{N_{\mathrm{out}}n}$.

\paragraph{Concatenated Code}
The concatenated code $\mathcal{C}_{\mathrm{con}}$ has total blocklength
\[
N_{\mathrm{con}} = n N_{\mathrm{out}} = n(q-1).
\]
Its minimum Hamming distance is
\[
D_{\mathrm{con}} = d_{\mathrm{in}} D_{\mathrm{out}}
= n\alpha(1-S-\epsilon)(q-K).
\]
The achievable rate is the product of the inner and outer rates,
{\small
\begin{align*}
R_{\mathrm{con}}
&= R'_{\mathrm{in}}R_{\mathrm{out}} = R'_{\mathrm{in}}\frac{K}{q-1} 
\end{align*}
}
By choosing $R_{\mathrm{out}}$ arbitrarily close to one, the concatenated construction preserves the positive rate guarantee of the inner code for sufficiently large blocklengths. For efficient implementation, we adopt a scaling rule for the inner code blocklength. For a target total blocklength $N_{\mathrm{con}}$, we take the inner code blocklength to be $n = c_0 \log N_{\mathrm{con}}$ and the outer code blocklength to be $N_{\mathrm{out}} = N_{\mathrm{con}}/n$. Under this scaling, the size of the inner codebook $\mathcal{C}_{\mathrm{in}}$ (and thus the alphabet size $q$) is bounded as
\begin{equation}
q \le |\mathcal{C}_{\mathrm{ec}}| = 2^{n R_{\mathrm{in}}} = 2^{(c_0 \log N_{\mathrm{con}}) R_{\mathrm{in}}} = \text{poly}(N_{\mathrm{con}}).
\end{equation}
Thus, inner code's encoding/decoding tables are polynomial in size relative to the overall blocklength. Inner encoding and decoding are performed via table lookup and minimum-distance decoding over the inner codebook. Outer encoding and decoding are carried out using RS codes with polynomial-time complexity in the overall blocklength. Therefore, the concatenated construction admits polynomial-time encoding and decoding while amplifying the distance guarantees of the expurgated inner code.

\section{Conclusion}
\label{sec:conclusion}

We investigated the construction of \emph{weakly constrained codes}, which generalize classical constrained codes by specifying the frequencies of occurrence of patterns instead of forbidding specific patterns completely. These frequencies are specified through constraints induced by an $n$-integral stationary Markov chain on a primitive labeled directed graph.

We first presented an Eulerian cycle based construction that generates a weakly constrained codebook that asymptotically achieves the capacity of the constrained system. We then established, via expurgation, the existence of weakly constrained codes with linear minimum distance and strictly positive rates for finite blocklengths. We derived achievable rates via concentration bounds and the LDP. To overcome the exponential implementation complexity of the expurgated codebook, we proposed a concatenated construction employing the expurgated code as an inner code and an RS outer code. By adopting a specific scaling rule for the inner blocklength, the resulting scheme preserves the weak constraints, amplifies the minimum distance, and admits polynomial-time encoding and decoding, which is essential for practical implementation. Together, these results provide a framework for weakly constrained coding that jointly addresses capacity, error correction, and computational efficiency. In our ongoing work, we are working on improving the finite length rate of our expurgated code by better concentration bounds and optimized codebook selection. 

\section*{Acknowledgement}
We would like to thank Ronny Roth for providing useful pointers on Eulerian cycles at the initial stages of this work.

\end{document}